\documentclass[a4paper, reqno]{amsart}
\usepackage{a4,wasysym}
\usepackage{amssymb}

\usepackage{amsmath,graphicx, amssymb,color,ulem,bbm}

\newtheorem{theorem}{Theorem}

\newtheorem{proposition}[theorem]{Proposition}
\newtheorem{definition}[theorem]{Definition}
\newtheorem{conjecture}[theorem]{Conjecture}
\newcommand{\ud}{\mathrm{d}}

\begin{document}
%%-----------------------------
%%      the top matter
%%-----------------------------
\title{Net Reproduction Functions for Nonlinear Structured Population Models} % At most 5 thanks
\author{J\'{o}zsef Z. Farkas}\address{Division of Computing Science and Mathematics, University of Stirling, Stirling FK9 4LA,UK; e-mail: jozsef.farkas@stir.ac.uk}
\begin{abstract} The goal of this note is to present a general approach to define the net reproduction function for a large class of nonlinear physiologically structured population models. In particular, we are going to show that this can be achieved in a natural way by reformulating a nonlinear problem as a family of linear ones; each of the linear problems describing the evolution of the population in a different, but constant environment. The reformulation of a nonlinear population model as a family of linear ones is a new approach, and provides an elegant way to study qualitative questions, for example the existence of positive steady states. To define the net reproduction number for any fixed (constant) environment, i.e. for the linear models, we use a fairly recent spectral theoretic result, which characterizes the connection between the spectral bound of an unbounded operator and the spectral radius of a corresponding bounded operator. For nonlinear models, varying the environment naturally leads to a net reproduction function. \end{abstract}
%
%\begin{resume} ... \end{resume}
%
\maketitle
%%-----------------------------
%%      your text
%%-----------------------------
\section{Prologue}

Physiologically structured population models have been developed and studied extensively in the past decades, see for example the monographs \cite{CUS,I,MD,W} for an in-depth introduction and review of the topic. Our ge\-ne\-ral research agenda, to relate the stability of steady states of physiologically structured population models to biologically meaningful quantities, allowed us to obtain a number of interesting results, see e.g. \cite{F,FH,FH2,FH4,FH3}. In particular, as we have seen for example in \cite{AF,F,FH,FH2,FH4,FH3}, the existence and stability of steady states of nonlinear structured population models can often be characterised using an appropriately defined net reproduction function. In fact we note that stability questions of nonlinear matrix population models were already investigated in the same spirit, see for example the paper \cite{YicangCushing}. 
However, we would like to emphasize that previously net reproduction functions were defined for concrete nonlinear models on an `ad hoc' basis, but typically via analysing the corresponding steady state equations, see e.g. \cite{AF,F,FGH,FH,FH2,FH4,FH3} for more details. It is our goal in this paper to present a general framework, which is applicable to different classes of nonlinear models.

The existence of positive steady states is an important and interesting question, when studying nonlinear po\-pu\-la\-tion models; and to establish the existence of non-trivial steady states is often challenging. This is especially the case for models formulated as infinite dimensional dynamical systems, for example delay equations, partial differential equations or integro-differential equations, see e.g. \cite{FarkasMorozov, FarkasMorozov2}. To overcome some of the difficulties, we have developed a very general framework to treat steady state problems of some classes of nonlinear partial differential or partial integro-differential equations, see \cite{CF,CF2,CF3}. The power of the method we developed becomes apparent for models with so-called infinite dimensional nonlinearities, when in fact the existence of positive steady states cannot be completely characterised by an appropriately defined net reproduction function. More precisely, for models with infinite dimensional nonlinearities, it turns out (as we will also see later) that one can still define a net reproduction functional $\mathcal{R}$, which will necessarily take the value $1$ at any positive steady state; but the condition $\mathcal{R}(E_*)=1$, where $E_*$ stands for the steady environment (for the concept of environment see e.g. \cite{DGM}), does not necessarily suffice to guarantee the existence of a positive steady state. This phenomenon was already noted for example in \cite{FH3} for a hierarchic size-structured population model (incorporating infinite dimensional nonlinearity). It also turns out, as we are going to show later, that the approach we developed in \cite{CF2} to study the existence of steady states of nonlinear models, can be utilised to define the net reproduction function/functional for a variety of nonlinear population models. We also note that our approach to consider a nonlinear population model, and then recast it in the form of a parametrised family of linear problems, is in some sense the opposite of the framework elaborated in detail for example in \cite{DGHKMT}, wherein the starting point is  a linear population model and then nonlinearities are incorporated via interaction variables. 

For linear models the basic (or net) reproduction number $\mathcal{R}$ could be defined (heuristically, or from the biological point of view) as the average number of offspring 'produced' by a newborn individual in her expected lifetime. Note that, crucially, this definition assumes that the environment any given individual is experiencing remains constant throughout her lifetime. This is the case in a linear model, and hence we can talk about a universal constant, the net reproduction number, 
which then often determines the asymptotic behaviour of solutions, too. 

As a motivating example consider the following linear age-structured population model, see e.g. \cite{I,W}.
\begin{align}
u_t(a,t)+u_a(a,t)= & -\mu(a)u(a,t),\quad a\in (0,m), \label{linage1} \\
u(0,t)=& \int_0^m \beta(a)u(a,t)\,\ud a.  \label{linage2}
\end{align}
For the basic age-structured model above one can introduce the so-called survival pro\-ba\-bility, as
\begin{align}
\pi(a)=\exp\left\{-\int_0^a\mu(r)\,\ud r\right\},\quad a\in [0,m].
\end{align}
Then one can argue that intuitively it should be clear that the average number of offspring `produced' by a newborn individual in her expected lifetime is 
\begin{align}\label{agerep}
\mathcal{R}:=\int_0^m\beta(a)\pi(a)\,\ud a.
\end{align}
The natural question arises, whether we can do this in a mathematically rigorous fashion; that is, to arrive at $\mathcal{R}$ as defined above in \eqref{agerep}, using a rigorous mathematical framework. The answer is yes, but surprisingly enough it turns out that the mathematical machinery what one has to deploy to this end is rather involved; it requires us to lift the age-structured problem \eqref{linage1}-\eqref{linage2} into a bigger (than the natural) state space, in which the boundary condition becomes a positive perturbation of the main part of the generator. In the forthcoming work \cite{CF4} we present the details of this calculation (amongst other things). Interestingly, it turns out that the mathematical machinery is much less involved/cumbersome for models incorporating infinite states at birth, as we will see later. It is also interesting to note this phenomenon because in the alternative delay formulation of physiologically structured population models, see for example \cite{D1}, this is usually the other way around; i.e. models with infinite states at birth are notoriously difficult to handle/analyse as the book-keeping process (tracing back individuals' history) becomes intractable. It may be also worthwhile to note that possibly the first account of the notion of a net reproduction number is found in \cite{DublinLotka}, where Dublin and Lotka were trying to establish a systematic way to estimate the growth rate of the (age-structured) human population.

In more realistic models than \eqref{linage1}-\eqref{linage2}, the environment that individuals experience changes, for example due to the availability of resources, which may be affected by competition (e.g. by scramble or contest competition). Hence for nonlinear models, one ought to define a net reproduction function or functional, analogously. Then, one can evaluate such functions or functionals by fixing the environment, to obtain scalar quantities, i.e. net reproduction numbers. Hence we would like to emphasize that for nonlinear models, one has to really talk about a net reproduction function/functional, analogous to \eqref{agerep}, but depending on the environment (interaction variable(s)). Then $\mathcal{R}_0$ for example is simply obtained by evaluating such a function/functional at a fixed environment (for example at the extinction environment, i.e. we may define $\mathcal{R}_0:=\mathcal{R}|_{E(0)}$). In fact, as demonstrated in the papers \cite{FH,FH2,FH4,FH3} for different types of  models, the stability of non-trivial steady states can be related to the derivative of the net reproduction function, while the stability of the extinction steady state is related to $\mathcal{R}_0$. Moreover, as shown for example by Cushing and later also by Walker, the  value $\mathcal{R}_0$ (if it is defined as $\mathcal{R}_0:=\mathcal{R}|_{E(0)}$) can be used as a bifurcation parameter to establish the existence of positive steady states, see for example \cite{C85,C85-2,W1,W2,W3}.

The very natural idea of defining a scalar quantity, namely the net reproduction number $\mathcal{R}$, when fixing the environment, can be made mathematically rigorous, and at the same time allows us to derive some desirable qualitative properties of the model. Note that the notation $\mathcal{R}_0$ itself (although it may well originate from \cite{DublinLotka}, but it has been adapted widely since then in the mathematical epidemiology community), possibly reflects the fact that mathematical epidemiologists for example may think about this number as (an approximation of) the number of secondary cases produced by a single infected individual in a completely susceptible population. This can be understood exactly in the sense as fixing the environment an infected individual experiences; for example (in a simple $SIS$ model) assuming that everybody else is, and will remain, susceptible. Importantly note that, by fixing a different environment a newly infected individual experiences, assuming for example that approximately half of the population is (and will remain) susceptible, will arguably lead to a different value of $\mathcal{R}_0$; and indeed we should really define a function $\mathcal{R}(S)$, see Section 4 for a concrete model to illustrate this. So we emphasize again, that also for nonlinear models of mathematical epidemiology one can  define a function or functional depending on the dimension of the nonlinearity incorporated in the model. In the context of population dynamical models, we argue that it seems reasonable to set $\mathcal{R}_0:=\mathcal{R}|_{E(0)}$, but of course there may be also different choices of parametrisations leading to different operators $E$, see also the subsequent sections for more details.

We also mention that in recent years there has been a reinvigorated interest within the mathematical epidemiology community and beyond to define net reproduction numbers also for linear non-autonomous models, for instance for models with explicitly time-dependent pa\-ra\-me\-ters (e.g. periodic vital rates). The interested reader will find many results in this direction for example in the fairly recent papers \cite{B0,B1,B2,Inaba}. It turns out that the general idea we present here can be applied to some linear non-autonomous models, see Section 4 for an illustrative example; but the information we may infer from the time-dependent net reproduction function we define is rather limited.

\section{(Re)formulation of a nonlinear problem}

Our starting point is a nonlinear evolution equation, which we assume to describe the dynamics of a physiologically structured population, formulated as a Cauchy problem in the Banach space $\mathcal{X}$ as follows.
\begin{equation}\label{nonlinproblem}
\frac{\ud u}{\ud t}=\mathcal{A}\,u,\quad \text{D}(\mathcal{A})\subseteq\mathcal{X},\quad \quad u(0)=u_0; 
\end{equation}
where $\mathcal{A}$ (possibly multivalued) is assumed to generate a (nonlinear) strongly continuous semigroup of type $\omega$, 
on $\mathcal{X}$ (or possibly on a closed subset of $\mathcal{X}$); that is, we tacitly assume in the rest of the paper that the assumptions of the Crandall-Liggett Theorem hold true, see \cite{CL}. In other words, we assume that (at least) mild (semigroup) solutions to \eqref{nonlinproblem} exist in $\mathcal{X}$ for all times. For some of the other notions not introduced here explicitly we refer the interested reader to \cite{AGG,NAG,Sch}. 

The (new) idea to recast the nonlinear problem \eqref{nonlinproblem} as a family of linear problems proved to be fruitful to establish results concerning the existence of steady states of model \eqref{nonlinproblem}, see \cite{CF,CF2,CF3}. Here, we apply the same formulation but in a slightly more specific fashion, to show how this idea can be also used to define a net reproduction function/functional associated to the nonlinear model \eqref{nonlinproblem} (of course when it describes the evolution of a biological population). In particular we assume that $\mathcal{X},\mathcal{Y}$ are Banach lattices, $\mathcal{X}$ being the state space; while we call $\mathcal{Y}$ the parameter space. $\mathcal{Y}$ can be understood as the set of all of the possible environments an individual may experience throughout her/his lifetime. From the mathematical point of view, typically there is a natural choice of the parameter space $\mathcal{Y}$, depending on the type of nonlinearity incorporated in model \eqref{nonlinproblem}.
 
We then recast the nonlinear model \eqref{nonlinproblem} as follows.
\begin{equation}\label{problem}
\frac{\ud u}{\ud t}=\mathcal{A}_{\bf p}\, u=\left(\mathcal{B}_{\bf p}+\mathcal{C}_{\bf p}\right)\,u,\quad \text{D}(\mathcal{B}_{\bf p}+\mathcal{C}_{\bf p})\subseteq\mathcal{X},\quad \quad u(0)=u_0, \quad {\bf p}\in\mathcal{Y}.
\end{equation}
Note that the state variable $u$ describes the density distribution of individuals. Since we have fixed the environment, for every ${\bf p}\in \mathcal{Y}$, $\mathcal{B}_{\bf p}$ and $\mathcal{C}_{\bf p}$ are linear operators. We split the operator $\mathcal{A}_{\bf p}$ in such a way that $\mathcal{B}_{\bf p}$ describes mortality and individual development, such as growth, and $\mathcal{C}_{\bf p}$ describes recruitment/reproduction of individuals. Note that, mortality and fertility being some of the basic ingredients of most population dynamical models, it seems natural from the biological point of view to distinguish the operators $\mathcal{B}_{\bf p}$ and $\mathcal{C}_{\bf p}$ in this way. 

At the same time we acknowledge, that of course there are models, even na\-tu\-ral ones, for which there can be an ambiguity as for example what can be considered recruitment or reproduction, as discussed for example in the most recent papers \cite{BCR,CD}; but we shall not be focusing on such models here. 
Rather we note that of course different splittings of $\mathcal{A}_{\bf p}$ may also result in significant technical advantages, especially in the context of mathematical epidemiology. For example in case of an epidemiological model different splittings of $\mathcal{A}$ may yield different formulas for the basic reproduction number $\mathcal{R}_0$, some of which may be easier to compute due for example to the availability of different types of patient data. See for example the recent paper \cite{BW} in this direction. 

From the mathematical point of view it is intuitively clear that if the operator $\mathcal{B}_{\bf p}$ describes mortality and individual development, then it will naturally have a negative spectral bound; while the recruitment operator $\mathcal{C}_{\bf p}$ is a positive operator, by definition (of course, when fixing the 'natural' positive cone of the state space). Moreover, in the absence of mortality, it seems reasonable to assume that the spectral bound of $\mathcal{B}_{\bf p}$ is zero. That is, for example jumps/transitions between individual states, such as the ones in the interesting example presented in Section 2 of \cite{CD}, should be naturally incorporated into $\mathcal{B}_{\bf p}$. On the one hand, a different splitting  may seem  counter-intuitive, on the other hand it may also destroy the desirable mathematical properties of the operators $\mathcal{B}_{\bf p}$ and $\mathcal{C}_{\bf p}$.

There is a natural (and important) relationship between elements of the parameter space ${\bf p}\in\mathcal{Y}$, and elements  $u\in\mathcal{X}$ of the state space. The relationship is determined by the so-called environmental operator: 
\begin{equation}\label{envcond}
\mathit{E}\,:\,\mathcal{X}\to\mathcal{Y},\quad \mathit{E}(u)={\bf p}.
\end{equation} 
From the biological point of view, $\mathit{E}$ typically determines how the standing population (and in turn the environment)  affects individual mortality,  development and fertility, hence the terminology. In general, the environment may be understood in the way that individuals in the population are independent of one another, when the environment is fixed (kept as constant), see \cite{DGM}; in which case the model becomes a linear one as in \eqref{problem}.  From the mathematical point of view, $\mathit{E}$ is in general a nonlinear operator (although in many of the concrete applications problem \eqref{problem} can be set up in such a way that $E$ is in fact positive linear, which has its advantages), associated with the type of nonlinearity incorporated in \eqref{nonlinproblem}. In particular, if it is possible to recast the nonlinear problem \eqref{nonlinproblem} in the form of \eqref{problem} with a particular choice of $\mathit{E}$, such that its range is contained in $\mathbb{R}^n$ for some $n\in\mathbb{N}$, then we say that problem \eqref{nonlinproblem} incorporates (finite) $n$-dimensional nonlinearity. On the other hand, if naturally $\mathcal{Y}$ itself is an infinite dimensional vector space, then we say that problem \eqref{nonlinproblem} incorporates an infinite dimensional nonlinearity.

To illustrate the parametrisation and splitting in \eqref{problem} through a concrete example, we consider the following  structured distributed states at birth model, see e.g. \cite{AF,CDF,FGH}. 
\begin{equation}\label{distr-eq} 
\begin{aligned}
u_t(s,t)+\left(\gamma(s,P(t))u(s,t)\right)_s&=-\mu(s,P(t))u(s,t)+\int_{s_{\text{min}}}^{s_{\text{max}}}\beta(s,\sigma,P(t))u(\sigma,t)\,\ud \sigma,\quad s\in(s_{\text{min}},s_{\text{max}}), \\
\gamma(s_{\text{min}},P(t))\,u(s_{\text{min}},t)&=0,\quad P(t)=\int_{s_{\text{min}}}^{s_{\text{max}}} u(s,t)\,\ud s.
\end{aligned}
\end{equation}
This model describes the time-evolution of a structured population via the population density function $u$. As in \cite{AF,CDF,FGH} we set the state space $\mathcal{X}=L^1\left(s_{\text{min}},s_{\text{max}}\right)$, the Banach space of equivalence classes of Lebesgue integrable functions defined on the interval $\left(s_{\text{min}},s_{\text{max}}\right)$. When compared to the classic age-structured model \eqref{linage1}-\eqref{linage2}, model \eqref{distr-eq} has the distinct feature that individuals may be recruited into the population at any possible state (size), i.e. it incorporates a distributed states at birth process (and hence the zero flux boundary condition in \eqref{distr-eq}). Also note that nonlinearities, in this case induced by scramble competition effects, are introduced via density dependence through in\-di\-vi\-dual mortality $\mu$, fertility $\beta$, and growth/development $\gamma$ rates. (We do not specify here the necessary regularity assumptions on the model ingredients, but instead refer the interested reader to \cite{AF,CDF,FGH}.) In model \eqref{distr-eq}, the environment individuals are experiencing is simply determined by the total population size (scramble competition), and therefore we set
\begin{equation}
E(u)=\int_{s_{\text{min}}}^{s_{\text{max}}} u(s)\,\ud s={\bf p}\in\mathbb{R}=\mathcal{Y},
\end{equation} 
and so in this example $E$ is a positive linear operator with range in $\mathbb{R}$, hence model \eqref{distr-eq} is said to  incorporate a one-dimensional nonlinearity. With this setting, the linear operators $\mathcal{B}_{\bf p},\,\mathcal{C}_{\bf p}$ can be naturally defined as follows (assuming that $\gamma$ and $\beta$ are sufficiently smooth, for example it suffices to assume that they are continuously differentiable). 
\begin{equation}\label{operatorC}
\begin{aligned}
\mathcal{B}_{\bf p}\,u= & -\left(\gamma(\cdot,{\bf p})u\right)'-\mu(\cdot,{\bf p})u, \quad \text{D}\left(\mathcal{B}_{\bf p}\right)=\left\{u\in W^{1,1}\left(s_{\text{min}},s_{\text{max}}\right)\,|\,u(s_{\text{min}})=0\right\},   \\
\mathcal{C}_{\bf p}\,u= &\displaystyle\int_{s_{\text{min}}}^{s_{\text{max}}}\beta(\cdot,y,{\bf p})u(y)\,\ud y, \quad \text{D}\left(\mathcal{C}_{\bf p}\right)=\mathcal{X}. 
\end{aligned}
\end{equation}
Note that in the natural splitting above, $\mathcal{B}_{\bf p}$ determines individual development (growth) and mortality (death), while $\mathcal{C}_{\bf p}$ determines the recruitment of individuals into the population (birth).

\section{The definition and basic properties of the net reproduction function}

Since for every fixed environment ${\bf p}\in\mathcal{Y}$ \eqref{problem} is a linear problem, we can define the net reproduction number $\mathcal{R}$ corresponding to each of such fixed environments. Since in principle we can do this for every ${\bf p}\in\mathcal{Y}$, we can write $\mathcal{R}({\bf p})$, with $\mathcal{R}\,:\,\mathcal{Y}\to\mathbb{R}_+$, and therefore introduce the notion of a net reproduction function/functional, depending on the dimension of the parameter space $\mathcal{Y}$. We recall from \cite{HT} Thieme's theorem, which naturally gives rise to define $\mathcal{R}({\bf p})$.
\begin{theorem}\label{spectralth} \cite[Theorem 3.5]{HT}
Let $\mathcal{B}$ be a resolvent-positive operator on $\mathcal{X}$, $s(\mathcal{B})<0$, and $\mathcal{A}=\mathcal{B}+\mathcal{C}$ a positive perturbation of $\mathcal{B}$. If $\mathcal{A}$ is resolvent-positive then \begin{equation}
s(\mathcal{A})\gtrless 0 \iff r\left(-\mathcal{C}\,\mathcal{B}^{-1}\right)-1 \gtrless 0.
\end{equation}
\end{theorem}
Above $s(\mathcal{A})$ and $s(\mathcal{B})$ stand for the spectral bound of the operators $\mathcal{A}$ and $\mathcal{B}$, respectively; while $r\left(-\mathcal{C}\,\mathcal{B}^{-1}\right)$ denotes the spectral radius of the bounded operator $-\mathcal{C}\,\mathcal{B}^{-1}$. 
Note that (at least in our opinion) the beauty of Theorem \ref{spectralth} is that its hypotheses exactly correspond to the biologically natural splitting of the generator $\mathcal{A}_{\bf p}$. If  \eqref{problem} describes a family of linear population models, then (for each fixed ${\bf p}$) it is naturally governed by a positive semigroup, with its generator $\mathcal{A}_{\bf p}=\mathcal{B}_{\bf p}+\mathcal{C}_{\bf p}$ being resolvent positive according to Definition 3.1 in \cite{HT}. Similarly, it can be shown that $\mathcal{B}_{\bf p}$ is resolvent positive, and since it incorporates mortality and individual development (but not recruitment) its spectral bound is necessarily negative. Finally $\mathcal{C}_{\bf p}$ describing recruitment (only) is necessarily a positive perturbation of $\mathcal{B}_{\bf p}$ as long as it is defined (at least) on the domain of $\mathcal{B}_{\bf p}$. We note that, typically $\mathcal{C}_{\bf p}$ can be defined on a larger set than the domain of $\mathcal{B}_{\bf p}$ (depending of course on the regularity assumptions we impose on the model ingredients), for example in case of the distributed states at birth model \eqref{distr-eq}, see \eqref{operatorC} above and also  \cite{AF,FGH} for more details.

We note at this point, that for example if mortality is large enough, then of course it would be possible to incorporate part of the  recruitment operator $\mathcal{C}_{\bf p}$ into $\mathcal{B}_{\bf p}$, while preserving all of the mathematical properties of these operators, which are necessary to apply Theorem \ref{spectralth}. But would such a splitting of $\mathcal{A}_{\bf p}$ be rather unnatural? 
Of course there could be a greater ambiguity how to split individual development/transitioning (if any) as discussed 
for example in \cite{CD}.

We now apply Theorem \ref{spectralth} at every fixed environment ${\bf p}$ to arrive at the definition of the net reproduction function. We tacitly assume that the operators in \eqref{problem} satisfy the assumptions of Theorem \ref{spectralth}. 
\begin{definition}
The net reproduction function for model \eqref{problem} is defined as
\begin{equation}\label{netrep}
\mathcal{R}({\bf p}):=r\left(-\mathcal{C}_{\bf p}\,\mathcal{B}_{\bf p}^{-1}\right).
\end{equation}
\end{definition}
Note that although this definition is mathematically inspired, the idea of defining the net reproduction number by fixing 
the environment (which, in many of the examples, amounts to keep the standing population size in the vital rates constant, see e.g. \eqref{operatorC}) seems to be rather natural. We also note that of course for concrete applications it may be very difficult (if not impossible) to compute the operator $-\mathcal{C}_{\bf p}\,\mathcal{B}_{\bf p}^{-1}$ (this was done for model \eqref{distr-eq} in \cite{AF}), and then to determine its spectral radius (again this was computed for the distributed states at birth model \eqref{distr-eq} in some special cases in \cite{AF}). For the age-structured (single state at birth) model \eqref{linage1}-\eqref{linage2} this is carried out in the forthcoming work \cite{CF4}.

Next we highlight that an immediate and desirable consequence of the definition of the net reproduction function above is that in general it implies a strong connection between the existence of positive steady states $u_*$ of the nonlinear problem  \eqref{nonlinproblem} and parameter values ${\bf p}_*$, such that $\mathcal{R}({\bf p}_*)=1$ holds. 
This is a very desirable property from the biological point of view, as one would naturally expect that non-trivial steady states of a population model arise when the net reproduction number equals $1$. For a large class of population models this is rigorously established as follows. 

\begin{proposition}\label{steadystate}
Assume that $\mathcal{B}_{\bf p}+\mathcal{C}_{\bf p}$ generates a positive, irreducible and eventually compact semigroup on $\mathcal{X}$ for all ${\bf p}\in\mathcal{Y}$. Then, if equation \eqref{nonlinproblem} admits a strictly positive steady state $u_*$, then $\mathcal{R}({\bf p_*})=1$, with ${\bf p}_*=E(u_*)$ holds.
\end{proposition}
\begin{proof}
If $u_*$ is a strictly positive positive steady state of \eqref{nonlinproblem}, then we have $u_*\in \text{Ker}\left(\mathcal{B}_{E(u_*)}+\mathcal{C}_{E(u_*)}\right)$, that is $0$ is an eigenvalue of the operator $\left(\mathcal{B}_{E(u_*)}+\mathcal{C}_{E(u_*)}\right)$. If however, the semigroup generated by $\left(\mathcal{B}_{E(u_*)}
+\mathcal{C}_{E(u_*)}\right)$ is irreducible and eventually compact, then the spectral bound is the only eigenvalue, which admits a corresponding strictly positive eigenvector (see e.g. \cite{AGG,NAG,Sch} for details), and therefore we have $s\left(\mathcal{B}_{E(u_*)}+\mathcal{C}_{E(u_*)}\right)=0$, which by Theorem \ref{spectralth} implies that $\mathcal{R}({\bf p_*})=1$, where ${\bf p}_*=E(u_*)\in\mathcal{Y}$. 
\end{proof}

Note that if in addition we assume that $E$ is positive (which is often the case in concrete applications, as seen e.g. in the previous section for model \eqref{distr-eq}), then we can restrict ourselves to parameter values ${\bf p}\in\mathcal{Y}_+$.

Importantly, note that Proposition \ref{steadystate} states that $\mathcal{R}({\bf p_*})=1$ is a necessary condition for the existence of a positive steady state $u_*$ of the nonlinear problem \eqref{nonlinproblem}. But the condition $\mathcal{R}({\bf p_*})=1$ cannot always shown to be sufficient, as already noted for example in \cite{CF2,FH3}. However, for a large class of models with one-dimensional nonlinearities this condition can in fact shown to be sufficient, too.

To this end, we impose some further restrictions on the environmental operator $E$. 
In particular, we assume that $E$ is of the form $E=\hat{E}+r$, for some $r\in\mathbb{R}$; where $\hat{E}\,:\,\mathcal{X}\to\mathbb{R}$ is such that $\forall\,\alpha\in\mathbb{R},\, \forall\,x\in\mathcal{X}$ we have $\hat{E}(\alpha x)=\alpha\hat{E}(x)$, and $\hat{E}(x)\neq 0$, unless $x\equiv 0$. We may say that the operator $E$ is affine, but strictly speaking it is not as $\hat{E}$ is not necessarily linear.

\begin{proposition}\label{steadystate2}
Assume that $\mathcal{Y}$ is one-dimensional, and that for every ${\bf p}\in\mathcal{Y}$ the operators $\mathcal{B}_{\bf p},\,\mathcal{C}_{\bf p}$ satisfy the assumptions of Theorem \ref{spectralth}. Moreover assume that $E$ satisfies the hypothesis above. Then, if there exists a ${\bf p}_*\in\mathcal{Y}$ such that $\mathcal{R}({\bf p}_*)=1$, then the nonlinear model  \eqref{nonlinproblem} admits a steady state.
\end{proposition}
\begin{proof}
If there exists a ${\bf p}_*\in\mathcal{Y}$ such that $\mathcal{R}({\bf p}_*)=1$, then by Theorem \ref{spectralth} we have that 
$s\left(\mathcal{B}_{\bf p_*}+\mathcal{C}_{\bf p_*}\right)=0$, with a subspace $V_{{\bf p}_*}\subseteq\mathcal{X}$ corresponding to the spectral value $0$, with dim$(V_{{\bf p}_*})\ge 1$, in general.  Note that we only need to show that there exists a $v_*\in V_{{\bf p}_*}$ such that $E(v_*)={\bf p}_*$ holds. However, it follows from our assumptions on $E$, that for every $0\not\equiv x\in\mathcal{X}$ 
the equation
\begin{equation}
\alpha\,\hat{E}(x)+r={\bf p}_*
\end{equation}
has a solution $\alpha\in\mathbb{R}$.
\end{proof}

Note that to allow a greater generality (than for example in \cite{CF2}) we did not the assert the existence of a strictly positive steady state in Proposition \ref{steadystate2}. Imposing further assumptions on the family of operators  $\mathcal{B}_{\bf p}+\mathcal{C}_{\bf p}$ (such that they generate eventually compact and irreducible semigroups), and on $E$, would allow us to obtain more information about the eigenspace $V_{{\bf p}_*}$, and in turn guarantee for example the existence of a strictly positive steady state.

We proceed by formulating two conjectures, which further underpin the significance of the definition of the net reproduction function in \eqref{netrep}, and in more general the idea of reformulating a nonlinear evolution equation \eqref{nonlinproblem} as a parametrised family of linear problems \eqref{problem}. Below by linearisation of \eqref{nonlinproblem} we mean the linear problem defined by using the linear approximation of the nonlinear operator $\mathcal{A}$ via its Fr\'{e}chet derivative (if it exists). We also note that in fact this assumption could be weakened by assuming only that $\mathcal{A}$ is G\^{a}teaux differentiable in the direction of the positive cone of $\mathcal{X}$.

\begin{conjecture}
Assume that the linearisation of \eqref{nonlinproblem} at the trivial steady state $u_*\equiv 0$ exists, and that $E(0)={\bf 0}$. Then, the trivial steady state $u_*\equiv 0$ is locally asymptotically stable if $\mathcal{R}({\bf 0})<1$; and it is unstable if $\mathcal{R}({\bf 0})>1$.
\end{conjecture}

This first conjecture can be interpreted as follows. The linearisation of \eqref{nonlinproblem} at the steady state $u_*\equiv 0$ (if it exists) coincides with the following linear Cauchy problem. 
\begin{equation}\label{linatzero}
\frac{\ud u}{\ud t}=\left(\mathcal{B}_{E(u_*\equiv 0)}+\mathcal{C}_{E(u_*\equiv 0)}\right)\,u,\quad \quad u(0)=u_0.
\end{equation}
This in fact can be verified directly for all of the structured population models (including models with single and with  infinite states at birth) we studied via linearisation techniques for example in \cite{F,FGH,FH,FH2,FH3}. 

The second conjecture we formulate may appear slightly far-fetched for the cautious reader, however we point out that there is a strong connection between the dissipativity condition of a generator of a linear (quasi-)contraction semigroup and the accretive condition of the generator of a nonlinear contraction semigroup. (For these basic concepts see e.g. \cite{CL,NAG}.)  Also note that here no linearisation of $\mathcal{A}$ is required, so if the claim holds, then it is potentially applicable to a wider class of models.

\begin{conjecture}
Assume  that $E(0)={\bf 0}$, and that there exists an $\varepsilon>0$ such that for every ${\bf p}\in\mathcal{Y}$, $\mathcal{R}({\bf p})<1-\varepsilon$ holds. Then the trivial steady state $u_*\equiv 0$ of \eqref{nonlinproblem} is globally asymptotically stable.
\end{conjecture}

The second conjecture may be interpreted as follows. If the net reproduction number is strictly less than $1$ for all possible environments (and in fact we have a uniform upper bound on it, which is less than $1$), then the population dies out, irrespective of the initial population density.

Finally we note that the ultimate (mathematical) advantage to recast and study the nonlinear model \eqref{nonlinproblem} in the form of a family of linear problems \eqref{problem} is that with the appropriate choice of the parameter space $\mathcal{Y}$,  and the environmental operator $E$, we may preserve the natural positivity properties of the original nonlinear problem, which are usually lost for example when approaching the nonlinear problem \eqref{nonlinproblem} with classical linearisation techniques. But we emphasize again that for a given nonlinear problem \eqref{nonlinproblem} typically there may be different ways to define the environmental operator $E$, which can then lead to different parametrisations \eqref{problem}, some of them  potentially more advantageous for concrete purposes.

\section{Further examples}

In this section we present further examples to illustrate the general ideas presented in the previous sections. In fact it is worth mentioning that the examples do not fit exactly into the framework presented in the previous sections. In particular first we consider an $SIS$ model, in which infected individuals are structured with respect to the bacterium/virus load they are carrying. We do not have a specific disease in mind, rather we consider a generic (toy model) to illustrate our approach. We note that a somewhat similar model, but with a super-linear term for the infectives, was analysed recently in \cite{CF3}. We consider the following coupled integro-differential equation model.
\begin{equation}\label{infection}
\begin{aligned}
i_t(x,t)+\left(\alpha(x)i(x,t)\right)_x &= -\varrho(x)i(x,t)+\int_0^\infty \Lambda(x,y)f(S(t))i(y,t)\,\ud y,\quad x\in(0,\infty), \\
S'(t) &=\int_0^\infty\varrho(x)i(x,t)\,\ud x-\int_0^\infty\int_0^\infty\Lambda(x,y)f(S(t))i(y,t)\,\ud y\,\ud x, \\
\alpha(0)i(0,t) &=0,\quad i(x,0)=i_0(x),\quad S(0)=S_0.
\end{aligned}
\end{equation}
In the model above $i$ stands for the density of infected individuals, while $S$ denotes the susceptible population. Furthermore, $\varrho$ denotes the recovery rate, $\alpha$ determines the rate at which bacteria/virus replicate inside an infected host individual thereby increasing their bacterial/viral load. $\Lambda$ determines the rate at which infected individuals pass on the disease to susceptibles upon contact. Newly infected individuals may have any bacterial/viral load in principle, hence we do not impose restrictions (other than regularity ones) on $\Lambda$. The function $f$ accounts for the fact that the number of contacts per unit time does not always linearly increase with the susceptible population size $S$. Note that in a standard  bilinear model we would simply have $f(S)\equiv S$, but in some applications it may be natural to assume that $f$ is sub or super-linear. It suffices to assume that the  model ingredients $\alpha,\varrho,\Lambda$ are positive and continuously differentiable, and also integrable on $(0,\infty)$. It is also natural to assume that $\alpha(\infty)=0$, and we assume that $f$ is positive and continuous. Note that there is no population dynamics incorporated in model \eqref{infection}, and indeed every solution $(i,S)$ satisfies 
\begin{equation*}
\frac{\ud}{\ud t} \left(\int_0^\infty i(x,t)\,\ud x+S(t)\right)=0,\quad t>0,
\end{equation*}
that is, the total population size remains constant for $t>0$. 
The natural choice of state space for model \eqref{infection} is $\mathcal{X}=L^1(0,\infty)\times\mathbb{R}$. 

It turns out that the most convenient parametrisation of \eqref{infection} arises when we define the environmental operator as 
\begin{equation}
E\,:\,\mathcal{X}\to\mathcal{Y}=\mathbb{R},\quad 
E \left(
\begin{array}{c}
i\\
S\\
\end{array}
\right)={\bf S}\in\mathbb{R}.
\end{equation} 
Note that previously we used ${\bf p}$ to denote the parameter, but here we emphasize that the parameter is actually the susceptible population size (and not some other function of $i$ and/or $S$), which is also a state variable, and hence we denote it by ${\bf S}$. 
We rewrite the nonlinear problem \eqref{infection} as the family of linear ones as 
\begin{equation}\label{infection-eq}
\frac{\ud i}{\ud t}=\left(\mathcal{B}+\mathcal{C}_{\bf S}\right)\,i,\quad i(0)=i_0,\quad {\bf S}\in\mathbb{R},
\end{equation} 
where we define the operators $\mathcal{B}$ and $\mathcal{C}_{\bf S}$ as follows
\begin{equation}\label{infection-param}
\begin{aligned}
\mathcal{B}\,i= & -\left(\alpha(\cdot)i\right)'-\varrho(\cdot)i, \quad \text{D}\left(\mathcal{B}\right)=\left\{i\in W^{1,1}\left(0,\infty\right)\,|\,i(0)=0\right\},   \\
\mathcal{C}_{\bf S}\,i= &f({\bf S})\,\int_0^\infty\Lambda(\cdot,y)i(y)\,\ud y, \quad \text{D}\left(\mathcal{C}_{\bf S}\right)=L^1(0,\infty). 
\end{aligned}
\end{equation}
That is, for any fixed value of ${\bf S}\in\mathbb{R}$, \eqref{infection-eq} is a (well-posed) linear model. 
Note that since the only nonlinearity in model \eqref{infection} arises due to the infection process, $\mathcal{B}$ naturally does not depend on ${\bf S}$. Also note that since the susceptible population size is chosen as the parameter ${\bf S}$, the second equation in \eqref{infection} becomes 'void' after the parametrisation. Note that in fact at a steady state the second equation of \eqref{infection} is simply the integral of the first one. It is not too difficult to verify that the operators $\mathcal{B}$ and $\mathcal{C}_{\bf s}$ satisfy the assumptions of Theorem \ref{spectralth} (see \cite{FGH} where such details are worked out for a similar model).  

With this setting, for any fixed ${\bf S}$, we can define the net reproduction number as the spectral radius of the (bounded) integral operator $\mathcal{L}_{\bf S}:=-\mathcal{C}_{\bf S}\,\mathcal{B}^{-1}$, which operator is given as 
\begin{equation}
\mathcal{L}_{\bf S}\,\phi = f({\bf S})\int_0^\infty\frac{\Lambda(\cdot,y)}{\alpha(y)}\int_0^y\exp\left\{-\int_r^y\frac{\varrho(\sigma)}{\alpha(\sigma)}\,\ud\sigma\right\}\phi(r)\,\ud r\,\ud y,\quad D(\mathcal{L}_{\bf S})=L^1(0,\infty).
\end{equation}
Therefore, we define the net reproduction function $\mathcal{R}$ (as a function of the susceptible population size) as
\begin{equation}
\mathcal{R}({\bf S}):=r\left(\mathcal{L}_{\bf S}\right).
\end{equation}
In this example the spectral radius of $\mathcal{L}_{\bf S}$, and therefore $\mathcal{R}$, behaves as $f$, as a function of the susceptible population size ${\bf S}$. For example in the classic bilinear case (i.e. when $f(S)\equiv S$) $\mathcal{R}$ is a monotone increasing function of ${\bf S}$, as expected. The spectral radius of $\mathcal{L}_{\bf S}$, and therefore $\mathcal{R}({\bf S})$, can be computed explicitly in some special cases, for example when $\Lambda$ is separable, in which case $\mathcal{L}_{\bf S}$ is of rank one. Also note that endemic steady states may arise at different values of ${\bf S}$, depending on the behaviour of the function $f$. At endemic steady states $(i_*,{\bf S}_*)$ we clearly have $\mathcal{R}({\bf S}_*)=1$, but in general $\mathcal{R}({\bf S})$ determines the expected number of secondary cases at any fixed susceptible population size ${\bf S}$. In fact the interested reader may recall from the celebrated paper \cite{DHM} the universal definition of the net reproduction number for infectious disease models. Indeed one may argue that it is already clear from \cite{DHM} (in particular see formulas (2.1)-(2.2) in \cite{DHM}), that $\mathcal{R}$ could be seen as a function of the variable ${\bf S}$.  

For the sake of (mainly mathematical) interest we mention that the general idea can be applied to some models incorporating explicitly time-dependent vital rates; although the information which can be inferred from the (time-dependent) net reproduction function is rather limited.
As an example consider the following linear, but non-autonomous version of the distributed states at birth model \eqref{distr-eq}. 
\begin{equation}\label{distr-eq-time} 
\begin{aligned}
u_t(s,t)+\left(\gamma(s,t)u(s,t)\right)_s&=-\mu(s,t)u(s,t)+\int_{s_{\text{min}}}^{s_{\text{max}}}\beta(s,\sigma,t)u(\sigma,t)\,\ud \sigma,\quad s\in(s_{\text{min}},s_{\text{max}}), \\
\gamma(s_{\text{min}},t)\,u(s_{\text{min}},t)&=0.
\end{aligned}
\end{equation}
It is clear that by fixing the time in the vital rates $\gamma,\beta,\mu$ we arrive at a parametrised family of linear evolution problems
\begin{equation}\label{distr-nonauto}
\frac{\ud u}{\ud t}=\left(\mathcal{B}_{\bf t}+\mathcal{C}_{\bf t}\right)\,u,\quad u(0)=u_0,\quad {\bf t}\in\mathbb{R},
\end{equation} 
where we define the linear operators $\mathcal{B}_{\bf t}$ and  $\mathcal{C}_{\bf t}$ as follows
\begin{equation}\label{nonauto}
\begin{aligned}
\mathcal{B}_{\bf t}\,u= & -\left(\gamma(\cdot,{\bf t})u\right)'-\mu(\cdot,{\bf t})u, \quad \text{D}\left(\mathcal{B}_{\bf t}\right)=\left\{u\in W^{1,1}\left(s_{\text{min}},s_{\text{max}}\right)\,|\,u(s_{\text{min}})=0\right\},   \\
\mathcal{C}_{\bf t}\,u= &\displaystyle\int_{s_{\text{min}}}^{s_{\text{max}}}\beta(\cdot,y,{\bf t})u(y)\,\ud y, \quad \text{D}\left(\mathcal{C}_{\bf t}\right)=\mathcal{X}. 
\end{aligned}
\end{equation}
It can be verified that under some natural assumptions on the model ingredients the operators $\mathcal{B}_{\bf t}$ and  $\mathcal{C}_{\bf t}$ satisfy the hypothesis of Theorem \ref{spectralth}. 
Similarly as earlier, for every fixed ${\bf t}$, the integral operator $\mathcal{L}_{\bf t}:=-\mathcal{C}_{\bf t}\,\mathcal{B}_{\bf t}^{-1}$ is given as
\begin{equation}
\mathcal{L}_{\bf t}\,\phi =\int_{s_{\text{min}}}^{s_{\text{max}}}\frac{\beta(\cdot,y,{\bf t})}{\gamma(y,{\bf t})}\int_0^y\exp\left\{-\int_r^y\frac{\mu(\sigma,{\bf t})}{\gamma(\sigma,{\bf t})}\,\ud\sigma\right\}\phi(r)\,\ud r\,\ud y,\quad D(\mathcal{L}_{\bf t})=L^1(s_{\text{min}},s_{\text{max}});
\end{equation}
and in turn we arrive at an explicitly time-dependent net reproduction function
\begin{equation}
\mathcal{R}({\bf t}):=r\left(\mathcal{L}_{\bf t}\right).
\end{equation}

\section{Perspectives}

In this note our goal was to show how to define the net reproduction function for wide classes of nonlinear structured population models. In particular, we have shown that this can be done by reformulating the original nonlinear problem as a family of linear ones, and then defining the net reproduction number for each of the linear problems using a recent spectral theoretic result from \cite{HT}. The unified approach we presented here to define the net reproduction function is important, as earlier results demonstrate that net reproduction functions play an important role in the qualitative analysis of structured population models, see e.g. \cite{FH,FGH,FH4,FH3} for further details.  

As we have already mentioned earlier, to actually compute the spectral radius of the appropriate integral operator, and to deduce an explicit formula for the net reproduction function is rather difficult, if not impossible, in general. For the infectious disease model \eqref{infection} we have indicated that it is possible to deduce an explicit formula for the case when the kernel function $\Lambda$ determining the rate of new infections is separable. 
The net reproduction function can be also computed for the distributed states at birth model, when the fertility function $\beta$ is separable, see e.g. \cite{AF,FGH} for more details. In general, numerical/computational tools may be employed, to obtain approximations of the function $\mathcal{R}$. For the distributed states at birth model \eqref{distr-eq} one possibility is to obtain a finite rank approximation of the operator $\mathcal{C}_{\bf p}$ defined in \eqref{operatorC}; for example by replacing the kernel function $\beta$ by the sum of finitely many separable ones.  We also refer the interested reader to \cite{Breda} for further details in this general direction.

On the other hand it is worthwhile to note that results on the continuity of a finite system of eigenvalues of closed operators (see e.g. Theorem 3.16 in \cite[Ch.IV]{Kato}) imply that in most cases (including all of the models we have studied so far) the net reproduction function $\mathcal{R}$, as defined in \eqref{netrep}, is a continuous function of its argument. 

Finally we mention that our focus here was on population dynamical models rather than on infectious disease models as in \cite{DHM}. 
Nevertheless we have provided an example of an SIS-type disease model \eqref{infection} to illustrate how our approach can be employed. 
For epidemiological models, deriving (or approximating) the net reproduction function (typically a function of the susceptible population size) may provide additional insight into the disease dynamics.

\section*{Acknowledgements}

I thank the Reviewer for several useful comments. I also thank \`Angel Calsina and Odo Diekmann for many insightful discussions.

%%-----------------------------
%%      your bibliography
%%-----------------------------
\end{document}